\documentclass[conference]{IEEEtran}
\usepackage{etex}
\usepackage{amsmath}
\usepackage{changepage,amsfonts}
\usepackage{amssymb}
\usepackage{thm-restate}
\usepackage{listings}
\usepackage{alltt}
\usepackage[usenames,dvipsnames]{color}
\usepackage{proof}
\usepackage{multirow}
\usepackage{graphicx}
\usepackage{enumitem}
\usepackage{wrapfig}
\usepackage{framed}
\usepackage[linesnumbered,algoruled,boxed,lined,noend]{algorithm2e}
\usepackage{amssymb, amsmath, amsthm, amsfonts}
\usepackage{tikz}
\usetikzlibrary{arrows,automata,positioning,calc}
\usepackage{adjustbox}
\usepackage[T1]{fontenc}
\tikzset{initial text={}}
\usetikzlibrary{decorations.markings}

\newcommand{\hide}[1]{}

\newcommand{\cA}{\mathcal{A}}
\newcommand{\cB}{\mathcal{B}}

\newcommand{\cT}{\mathcal{T}}
\newcommand{\cI}{\mathcal{I}}
\newcommand{\cR}{\mathcal{R}}

\newcommand{\langA}{A}
\newcommand{\langB}{B}
\newcommand{\langT}{T}
\newcommand{\langI}{I}
\newcommand{\langR}{R}
\newcommand{\langE}{E}

\newcommand\ltrue{\mathsf{true}}
\newcommand\lfalse{\mathsf{false}}

\newcommand{\MEM}{\mathit{Mem}}
\newcommand{\EQ}{\mathit{Equ}}

\newcommand{\libalf}[0]{\textsf{\small libalf}}
\newcommand{\T}[0]{\textsf{T}}
\newcommand{\N}[0]{\textsf{N}}

\newcommand{\OMIT}[1]{}
\newcommand{\Nat}{\mathbb{N}}

\usepackage{xcolor}
\definecolor{linkcolor}{rgb}{0.65,0,0}
\definecolor{citecolor}{rgb}{0,0.65,0}
\definecolor{urlcolor}{rgb}{0,0,0.65}
\usepackage[colorlinks=true, linkcolor=linkcolor, urlcolor=urlcolor, citecolor=citecolor]{hyperref}

\newtheorem{theorem}{Theorem}
\newtheorem{example}{Example}
\newtheorem{proposition}{Proposition}

\ifCLASSINFOpdf
\else
\fi
\begin{document}
%
% paper title
% Titles are generally capitalized except for words such as a, an, and, as,
% at, but, by, for, in, nor, of, on, or, the, to and up, which are usually
% not capitalized unless they are the first or last word of the title.
% Linebreaks \\ can be used within to get better formatting as desired.
% Do not put math or special symbols in the title.
\title{Learning to Prove Safety over Parameterised Concurrent Systems (Full
Version)}

% author names and affiliations
% use a multiple column layout for up to three different
% affiliations
\author{\IEEEauthorblockN{Yu-Fang Chen}
\IEEEauthorblockA{Academia Sinica}
\and
\IEEEauthorblockN{Chih-Duo Hong}
\IEEEauthorblockA{Oxford University}
\and
\IEEEauthorblockN{Anthony W. Lin}
\IEEEauthorblockA{Oxford University}
\and
\IEEEauthorblockN{Philipp R\"ummer}
\IEEEauthorblockA{Uppsala University}
}

% conference papers do not typically use \thanks and this command
% is locked out in conference mode. If really needed, such as for
% the acknowledgment of grants, issue a \IEEEoverridecommandlockouts
% after \documentclass

% for over three affiliations, or if they all won't fit within the width
% of the page, use this alternative format:
% 
%\author{\IEEEauthorblockN{Michael Shell\IEEEauthorrefmark{1},
%Homer Simpson\IEEEauthorrefmark{2},
%James Kirk\IEEEauthorrefmark{3}, 
%Montgomery Scott\IEEEauthorrefmark{3} and
%Eldon Tyrell\IEEEauthorrefmark{4}}
%\IEEEauthorblockA{\IEEEauthorrefmark{1}School of Electrical and Computer Engineering\\
%Georgia Institute of Technology,
%Atlanta, Georgia 30332--0250\\ Email: see http://www.michaelshell.org/contact.html}
%\IEEEauthorblockA{\IEEEauthorrefmark{2}Twentieth Century Fox, Springfield, USA\\
%Email: homer@thesimpsons.com}
%\IEEEauthorblockA{\IEEEauthorrefmark{3}Starfleet Academy, San Francisco, California 96678-2391\\
%Telephone: (800) 555--1212, Fax: (888) 555--1212}
%\IEEEauthorblockA{\IEEEauthorrefmark{4}Tyrell Inc., 123 Replicant Street, Los Angeles, California 90210--4321}}

% use for special paper notices
%\IEEEspecialpapernotice{(Invited Paper)}

% make the title area
\maketitle

% As a general rule, do not put math, special symbols or citations
% in the abstract
\begin{abstract}
 We revisit the classic problem of proving safety over 
parameterised concurrent systems, i.e., an infinite family of finite-state
concurrent systems that are represented by some finite (symbolic) means. 
An example of such an infinite family is a dining philosopher protocol with any 
number $n$ of processes ($n$ being the parameter that defines the infinite 
family). Regular model checking is a well-known generic framework for modelling 
parameterised concurrent systems, where an infinite set of configurations 
(resp.~transitions) is represented by a regular set (resp.~regular transducer). 
Although verifying safety properties in the regular model checking framework is undecidable in 
general, many sophisticated 
semi-algorithms have been developed in the past fifteen years that can 
successfully prove safety in many practical instances. 
In this paper, we propose a simple solution to synthesise regular inductive invariants that makes use of Angluin's classic $L^*$ algorithm (and its variants).
We provide a termination guarantee 
when the set of configurations reachable from a given set of initial
configurations is regular.
%--- a 
%conceptually simpler method --- 
We have tested $L^*$ algorithm on standard (as well as new) examples in regular 
model checking including the dining philosopher protocol, the dining cryptographer
protocol, and several mutual exclusion protocols (e.g. Bakery, Burns, Szymanski, 
and German). Our experiments show that, despite the simplicity of our solution, it can perform at least as well as
existing semi-algorithms. 

\end{abstract}

% For peer review papers, you can put extra information on the cover
% page as needed:
% \ifCLASSOPTIONpeerreview
% \begin{center} \bfseries EDICS Category: 3-BBND \end{center}
% \fi
%
% For peerreview papers, this IEEEtran command inserts a page break and
% creates the second title. It will be ignored for other modes.
\IEEEpeerreviewmaketitle

\section{Introduction}
\label{section:introduction}
Parameterised concurrent systems are infinite families of finite-state
concurrent systems, parameterised by the number $n$ of processes.
There are numerous examples of parameterised concurrent systems, including
models of distributed algorithms %(e.g. a dining philosopher protocol
which are typically designed to handle an arbitrary number $n$ of processes
\cite{Fokkink-book,LSS94}.
Verification of such systems, then, amounts to
proving that a desired property holds for \emph{all} permitted values of $n$.
For example, proving that the safety property holds for a dining philosopher
protocol entails proving that the
protocol with any given number~$n$ of philosophers ($n \geq 3$) can never reach
a state when two neighbouring philosophers eat simultaneously.
For each given value of $n$, verifying safety/liveness is decidable, albeit
the exponential state-space explosion in the parameter $n$. However, when
the property has to hold for each value of $n$, the number of system
configurations a verification algorithm has to explore is potentially infinite.
Indeed, even safety checking is already
undecidable for parameterised concurrent systems \cite{AK86,BF13,Javier16}; see
\cite{sasha-book}
for a comprehensive survey on the decidability aspect of the parameterised
verification problem.

Various sophisticated semi-algorithms for verifying parameterised concurrent
systems are available. These semi-algorithms typically rely on a
symbolic framework for representing infinite sets of system configurations
and transitions. \emph{Regular model checking}
\cite{JN00,rmc-survey,Boigelot-thesis,BLW03,LTL-MSO,Parosh12,RMC,rmc-pnueli,WB98}
is one well-known symbolic framework for modelling and verifying parameterised
concurrent systems. In regular model checking, configurations are modelled
using words over a finite alphabet, sets of configurations are
represented as regular languages, and the transition relation is defined by a
regular transducer. From the research programme of regular model checking,
not only are regular languages/transducers known to be highly expressive
symbolic representations for modelling parameterised concurrent systems,
they are also amenable to an automata-theoretic approach (due to many nice
closure properties of regular languages/transducers), which have often proven
effective in verification.

In this paper, we revisit the classic problem of verifying safety in the regular
model checking framework. Many sophisticated semi-algorithms for dealing with
this problem have been developed in the literature using methods such as
abstraction~\cite{abdulla2007regular,abdulla2013all,bouajjani2004abstract,bouajjani2006abstract},
widening~\cite{BLW03,BT-widening}, acceleration~\cite{rmc-thesis,JN00,BFLS05},
and learning~\cite{Neider-thesis,Neider13,habermehl2005regular,Lever,VSVA04}.
One standard technique for proving safety for an infinite-state systems is by
exhibiting an \emph{inductive invariant} $\mathit{Inv}$ (i.e. a set of configurations
that is closed under an application of the transition relation) such that
\textbf{(i)} $\mathit{Inv}$
subsumes the set $\mathit{Init}$ of all initial configurations, but
\textbf{(ii)} $\mathit{Inv}$ does not intersect with the set
$\mathit{Bad}$ of unsafe configurations.
%\textbf{Inductiveness should be (iii)?} YF:I think it is fine without (iii), I will echo to this by emphasising in the technical part that  the condition (iii) is inductiveness
In regular model checking, the sets $\mathit{Init}$
and $\mathit{Bad}$ are given as regular sets. For this reason, a natural method for
proving safety in regular model checking is to exhibit a \emph{regular}
inductive invariant satisfying (i) and (ii). The regular set
$\mathit{Inv}$ can be constructed as a ``regular proof'' for safety since checking that
a candidate regular set $\mathit{Inv}$ is a proof for safety is decidable.
A few semi-algorithms inspired by automata learning %in the sense of Angluin\cite{angluin1988queries,angluin1987learning}
--- some based on the passive learning algorithms \cite{habermehl2005regular,Neider13,AbdullaACHRRS14} and
some others based on active learning algorithms \cite{Neider13,VSVA04}--- have been proposed to synthesise a regular inductive invariant in regular model checking.
Despite these
semi-algorithms, not much attention has been paid to applications of automata
learning in regular model checking.

\OMIT{
Most properties of interests
in the verification of concurrent systems (including safety and liveness) hold
\emph{regardless of the concrete value of the parameter $n$}.
}

In this paper, we are interested in one basic research question in regular
model checking: \emph{can we effectively apply the classic Angluin's $L^*$
automata learning \cite{angluin1987learning} (or variants
\cite{rivest:inference1993,kearns:introduction1994}) to learn a regular
inductive invariant?} Hitherto this question, perhaps surprisingly, has no
satisfactory answer in the literature. A more careful consideration reveals
at least two problems.
Firstly, membership
queries (i.e.~is a word $w$ reachable from $\mathit{Init}$?) may be asked by the $L^*$
algorithm, which amounts to checking reachability in an infinite-state system,
which is undecidable in general. This problem was already noted in
\cite{Neider-thesis,Neider13,VSVA04,Lever}.
\OMIT{
Fortunately, there is an elegant solution to this issue, namely the
representation of systems in terms of \emph{length-preserving} regular
transducers, i.e., as relations that can only relate two words of the
same length.
}
 Secondly, a regular inductive invariant satisfying (i)
and (ii) might not be unique, and so strictly speaking we are not dealing with
a well-defined learning problem. More precisely, consider the question of
\emph{what the teacher should answer when the learner asks whether $v$ is
in the desired invariant, but $v$ turns out not to be reachable from $\mathit{Init}$?}
Discarding $v$ might not be a good idea, since this could force the learning algorithm to look for a
\emph{minimal} (in the sense of set inclusion) inductive invariant, which
might not be regular.
%that all \emph{regular} inductive invariants include $v$.
Similarly, let us consider what the teacher should answer in the case when we
found a pair $(v,w)$ of configurations such that (1) $v$ is in the candidate
$\mathit{Inv}$, (2) $w \notin \mathit{Inv}$, and (3) there is a transition from $v$
to $w$. In the ICE-learning framework \cite{GLMN13,GLMN14,Neider-thesis}, the
pair $(v,w)$ is called an \emph{implication counterexample}. To satisfy the
inductive invariant constraint, the teacher may respond that $w$ should be added to $\mathit{Inv}$,
or that $v$ should be removed from $\mathit{Inv}$.  Some works in
the literature have proposed using a three-valued logic/automaton (with 
``don't know'' as an answer) because of the
teacher's incomplete information \cite{grinchtein2006inferring,ChenFCTW09}.

\OMIT{
from which
there is a transition out of $\mathit{Inv}$? As noted in \cite{Neider-thesis}, the
answer to this question is not unique since an inductive invariant only
needs to be an over-approximation of the set $post^*(\mathit{Init})$ of configurations
reachable from $\mathit{Init}$ (a.k.a. the reachability set).
}

\paragraph{Contribution} In this paper, we propose a simple and practical
solution to the problem of applying the classic $L^*$ automata learning
algorithm and its variants to synthesise a regular inductive invariant in regular model
checking.
To deal with the first problem mentioned in the previous paragraph, we
propose to restrict to \emph{length-preserving} regular transducers.
In theory, length-preservation is not a restriction for safety
analysis, since it just implies that each instance of the considered
parameterised system is operating on bounded memory of size~$n$ (but
the parameter~$n$ is unbounded).  Experience shows that many practical
examples in parameterised concurrent systems can be captured naturally in terms
of length-preserving
systems, e.g., see
\cite{lin2016liveness,rmc-survey,LTL-MSO,JN00,RMC,rmc-thesis,Parosh12}. The
benefit of the restriction is that the problem of membership queries
%with respect to a strict teacher
is now decidable, since the set of
configurations that may reach (be reachable from) any given configuration $w$
is finite and can be solved by a standard finite-state model checker.
For the second problem mentioned in the previous paragraph, we propose
that a \emph{strict teacher} be employed in $L^*$
learning for regular inductive invariants in regular model checking. A
strict teacher attempts to teach the learner the minimal inductive invariant
(be it regular or not), but is satisfied when the candidate answer posed by
the learner is an inductive invariant satisfying (i) and (ii) without
being minimal. [In this sense, perhaps a more appropriate term is a
\emph{strict but generous teacher}, who tries to let a student pass a final exam
whenever possible.]
For this reason, when the learner asks whether $w$ is in the desired inductive
invariant, the teacher will reply NO if $w$ is not reachable from $\mathit{Init}$.
The same goes with an
implication counterexample $(v,w)$ such that the teacher will say that
an unreachable $v$ is not~in~the~desired~inductive~invariant.

We have implemented the learning-based approach in a prototype tool with an interface to the \libalf\ library, which includes the $L^*$ algorithm and its variants.
Despite the simplicity of our solution, it (perhaps surprisingly) works
extremely well in practice, as our
experiments suggest. We have taken numerous standard examples from regular
model checking, including cache coherence protocols (German's Protocol),
self-stabilising protocols (Israeli-Jalfon's Protocol and Herman's
Protocol), synchronisation protocols (Lehmann-Rabin's Dining Philosopher
Protocol), secure multi-party computation protocols (Dining Cryptographers
Protocol \cite{chaum1988dining}), and mutual exclusion
protocols (Szymanski's
Protocol, Burn's Protocol, Dijkstra's Protocol, Lamport's bakery algorithm,
and Resource-Allocator Protocol).
We show that $L^*$ algorithm can perform
at least as well as (and, in fact, often outperform) existing semi-algorithms.
We compared the performance of our algorithm with well-known and established  techniques
such as SAT-based learning~\cite{Neider13,Neider-thesis,lin2015regular,lin2016liveness}, abstract regular model checking (ARMC), which is based on abstraction-refinement using predicate abstractions and finite-length abstractions \cite{bouajjani2006abstract,bouajjani2004abstract}, and T(O)RMC, which is based on extrapolation (a widening technique)~\cite{tormc}.
Our experiments show that, despite the simplicity of our solution, it can perform at least as well as
existing semi-algorithms.

\paragraph{Related Work}
The work of Vardhan \emph{et al.} \cite{Lever,VSVA04} applies $L^*$ learning to
infinite-state systems and, amongst other, regular model checking. The
learning algorithm attempts to learn an inductive invariant enriched with
``distance'' information, which is one way to make membership queries (i.e. reachability for general infinite-state systems) decidable. This often makes the resulting set not regular, even
if the set of reachable configurations is regular, in which case our algorithm is guaranteed to terminate (recall our algorithm is only learning a regular invariant without distance information). Conversely, when an inductive invariant enriched with distance information is regular, so is the projection that omits the distance information.
Unfortunately, neither their tool Lever~\cite{Lever}, nor the
models used in their experiments are available, so that we cannot
make a direct comparison to our approach. A learning algorithm allowing incomplete information~\cite{grinchtein2006inferring} has been applied in~\cite{Neider13} for inferring inductive invariants of regular model checking. Although the learning algorithm in~\cite{grinchtein2006inferring} uses the same data structure as the standard $L^*$ algorithm, it is essentially a SAT-based learning algorithm (its termination is not guaranteed by the Myhill-Nerode theorem).

Despite our results that SAT-based learning seems to be less efficient than
$L^*$ learning for synthesising regular inductive invariants in regular model
checking, SAT-based learning is more general and more easily applicable
when verifying other properties, e.g., liveness \cite{lin2016liveness},
fair termination \cite{LLMR17}, and safety games \cite{NT16}.
View abstraction~\cite{abdulla2013all} is a novel technique for parameterised verification.
Comparing to parameterised verification based on view abstraction, our framework
(i.e. general regular model checking framework with transducers)
provides a more expressive modelling language that is required in specifying
protocols with near-neighbour communication (e.g. Dining Cryptographers and
Dining Philosophers).
%like Dining Cryptographers (one need to specify near-neighbour communication in the model).

When preparing the final version, we found that
a very similar algorithm had already appeared in Vardhan's
thesis~\cite[Section~6]{Vardhan-thesis} from 2006; in particular,
including the trick to make a membership query (i.e.\ point-to-point
reachability) decidable by bounding the space of the transducers.  The
research presented here was conducted independently, and considers
several aspects that were not yet present in \cite{Vardhan-thesis},
including experimental results on systems that are not counter systems
(parameterised concurrent systems with topologies), and heuristics
like the use of shortest counterexamples and caching. We cannot
compare our implementation in detail with the one from
\cite{Vardhan-thesis}, since the latter tool is not publicly
available.

\iffalse
We recently discovered that in his thesis \cite{Vardhan-thesis} Vardhan already
proposed how to make a membership query (i.e. point-to-point reachability) 
\cite{Vardhan-thesis} decidable by bounding the space of the transducers.
We rediscovered this trick without knowing the proposal in his thesis. 
We cannot compare our verification algorithm with his in detail since his tool 
is not publicly available. Thus, the following discussion is solely based on
what is described in \cite{Vardhan-thesis}. While his
experiments focus on transducers modelling integer programs,
 %with 
%experimental results on transducers modelling integer programs (counter
%systems). 
ours focus on parameterised concurrent systems with topologies. 
%Although the essence of the verification algorithms is similar, our algorithm
%also implements 
Our algorithm also includes several heuristics (e.g. shortest
counterexamples and caching) that was not described in \cite{Vardhan-thesis}.
\fi

\paragraph{Organisation}
The notations are defined in Section~\ref{section:preliminaries}. A brief introduction to regular model checking and automata learning is given in Section~\ref{section:rmc} and Section~\ref{section:learning}, respectively. The learning-based algorithm is provided in Section~\ref{section:algorithm}. The result of the experiments is in Section~\ref{section:evaluation}.
%% 2 pages

%%% Local Variables:
%%% mode: latex
%%% TeX-master: "draft"
%%% End:

%\vspace{-0.2cm}
\section{Preliminaries}
%\vspace{-0.2cm}
\label{section:preliminaries}

%% 3 pages

%% notes:
%% -Finite automata and Transducer
%% -Regular Model Checking and its Goal
%% -Automata Learning

\paragraph{General Notations}
Let $\Sigma$ be a finite set of symbols called \emph{alphabet}. A word over $\Sigma$ is a finite sequence
of symbols of $\Sigma$. We use $\lambda$ to represent an empty word.
For a set $I\subseteq \Sigma^*$ and a relation $T\subseteq \Sigma^*\times 
\Sigma^*$, we define $T(I)$ to be the post-image of $I$ under $T$, i.e.,
$T(I)=\{y\mid \exists x .\ x\in I \wedge (x,y)\in T\}$.
Let $id =\{(x,x)\mid x\in \Sigma^*\}$ be the {\em identity relation}. We define
$T^n$ for all $n \in \Nat$ in the standard way by induction: $T^0=id$, and
$T^k =T \circ T^{k-1}$, where $\circ$ denotes the {\em composition} of 
relations. Let $T^*$ denote the transitive closure of $T$, i.e., $T^* 
=\bigcup_{i=1}^\infty T^i$. For any two sets $\langA$ and $\langB$, we use 
$\langA \ominus \langB$ to denote their {\em symmetric difference}, i.e., the set $\langA\setminus \langB \cup \langB\setminus\langA$.

%We use $[i\cdots j]$ for the set $\{k \mid i \leq k\leq  j\}$. We denote by $w[i]$ the $i$-th letter of $w$ (started from $1$) and $w[i..k]$ for the subword of $v$ starting at the $i$-th letter and ending at the $k$-th letter, inclusive.

\paragraph{Finite Automata and Transducer}
In this paper, automata/transducers are denoted in calligraphic fonts 
$\cA,\cB,\cI,\cT$ to represent automata/transducers, while the corresponding
languages/relations are denoted in roman fonts $\langA,\langB,\langI,\langT$.

A {\em finite automaton} (FA) is a tuple $\cA=(Q,\Sigma, \delta, q_0, F)$ where
$Q$ is a finite set of states, $\Sigma$ is an alphabet, $\delta\subseteq Q\times
\Sigma\times Q$ is a transition relation, $q_0\in Q$ is the initial state, and
$F\subseteq Q$ is the set of final states. A {\em run} of $\cA$ on a word $w = a_1 a_2 a_3 \cdots a_n$ is a sequence of states $q_0, q_1, \cdots, q_n$
such that $(q_i, a_{i+1}, q_{i+1}) \in \delta$. A run is {\em accepting} if the
last state $q_n\in F$. A word is {\em accepted} by $\cA$ if it has an accepting
run. The {\em language} of $\cA$, denoted by $\langA$, is the set of word accepted by $\cA$. 
%When the context is clear, we abuse the notation and use $\cA$ to represent its language $\langA$.
%We use $|\cA|$ to denote the number of states in $\cA$.
A language is {\em regular} if it can be recognised by a finite automaton.
$\cA$ is a \emph{deterministic finite automaton} (DFA) if 
$|\{q'\mid (q, a,q')\in \delta\}| \leq 1$ for each $q\in Q$ and $a\in \Sigma$.

Let $\Sigma_\lambda=\Sigma\cup \{\lambda\}$.
A {\em (finite) transducer} is a tuple $\cT=(Q,\Sigma_\lambda, \delta, q_0, F)$ where $Q$ is a finite set of states, $\delta\subseteq Q\times\Sigma_\lambda\times \Sigma_\lambda\times Q$ is a transition relation, $q_0\in Q$ is the initial state, and $F\subseteq Q$ is the set of final states. 
We say that $\cT$ is {\em length-preserving} if $\delta \subseteq Q\times\Sigma\times \Sigma\times Q$. 
We define relation $\delta^* \subseteq Q\times \Sigma^*\times \Sigma^* \times Q$ as the smallest relation satisfying (1) $(q,\lambda,\lambda,q) \in \delta^*$ for any $q \in Q$ and (2) $(q_1,x,y,q_2) \in \delta^* \wedge (q_2,a,b,q_3) \in\delta \implies (q_1,xa,yb,q_3)\in \delta^*$. 
%When the context is clear, we abuse the notation and use $\cT$ to represent $\delta^*$, the relation identified by $\cT$. 
The relation represented by $\cT$ is the set $\{(x,y)\mid (q_0,x,y,q) \in \delta^* \wedge q\in F\}$.
A relation is {\em regular and length-preserving} if it can be represented by a length-preserving transducer. 

\vspace{-0.1cm}
\section{Regular model checking}\label{section:rmc}
\vspace{-0.1cm}
{\em Regular model checking} (RMC) is a uniform framework for modelling and
automatically analysing parameterised concurrent systems. In the paper, we focus
on the regular model checking framework for safety properties. Under the framework, each {\em system configuration} is represented as a word in $\Sigma^*$. The sets of {\em initial configurations} and of {\em bad configurations} are captured by regular languages over $\Sigma$. The {\em transition relation} is captured by a regular and length-preserving relation on $\Sigma^*$. We use a triple $(\cI,\cT,\cB)$ to denote a {\em regular model checking problem}, where $\cI$ is an FA recognizing the set of initial configurations, $\cT$ is a transducer representing the transition relation, and $\cB$ is an FA recognizing the set of bad configurations. 
Then the regular model checking problem $(\cI, \cT, \cB)$ asks if
$\langT^*(\langI) \cap \langB =\emptyset$. A standard way to prove
$\langT^*(\langI) \cap \langB =\emptyset$ is to find a proof based on a set $V$
satisfying the following three conditions:
(1)$\langI\subseteq V$ (i.e. all initial configurations are contained in $V$), 
(2) $V\cap \langB=\emptyset$ (i.e. $V$ does not contain bad
configurations), (3) $\langT(V) \subseteq V$ (i.e. $V$ is {\em inductive}:
applying $T$ to any configuration in $V$ does not take it outside $V$).
%(1) $\langI\subseteq V$, (2) $\langT(V) \subseteq V$, and (3) $V\cap \langB=\emptyset$. 
\OMIT{
\begin{itemize}
	\item $\langI\subseteq V$: all initial configurations are contained in $V$.
	\item $V\cap \langB=\emptyset$: $V$ is disjoint with the set of bad configurations.
	\item $\langT(V) \subseteq V$: $V$ is {\em inductive}, i.e., the successors of any configuration in $V$ following the transition relation $\langT$ are 
            still contained in $V$.
\end{itemize}}
We call the set $V$ an {\em inductive invariant} for the regular model checking
problem $(\cI,\cT,\cB)$. In the framework of regular model checking, 
a standard method for proving safety (e.g. see \cite{Neider13,rmc-survey})
is to find a {\em regular proof}, i.e., 
an inductive invariant that can be captured by finite automaton. Because regular
languages are effectively closed under Boolean operations and taking
pre-/post-images w.r.t.\ finite transducers, an algorithm for verifying  whether
a given regular language is an inductive invariant can be obtained by using
language inclusion algorithms for FA~\cite{abdulla-antichain,BP13}. 

\begin{example}[Herman's Protocol]\label{example:herman}
    Herman's Protocol is a {\em self-stabilising} protocol for $n$ processes
    (say with ids $0,\ldots,n-1$)
    organised as a ring structure. A {\em configuration} in the  Herman's Protocol is {\em correct} iff only one process has a token.
The protocol ensures that any system configuration where the processes collectively holding any odd number of tokens will almost surely be recovered to a correct configuration.
More concretely, the protocol works iteratively. In each iteration, the scheduler randomly chooses a process. 
If the process with the number $i$ is chosen by the scheduler, it will toss a coin to decide whether to keep the token or pass the token to the next process, i.e. the one with the number $(i+1)\%n$. If a process holds two tokens in the same iteration, it will discard \emph{both} tokens. One safety 
property the protocol guarantees is that every system configuration has at least one token.

The protocol and the corresponding safety property can be modelled as a regular model checking problem $(\cI, \cT, \cB)$. Each process has two states; the symbol \T\ denotes the state that the process has a token and \N\ denotes the state that the process does not have a token. The word \N\N\T\T\N\N\ denotes a system configuration with six processes, where only the processes with numbers $2$ and $3$ are in the state with tokens.
The set of initial configurations is $\langI=\N^*\T(\N^*\T\N^*\T\N^*)^*$, i.e.,
    an odd number of processes has tokens. The set of bad configuration is
    $\langB=\N^*$, i.e., all tokens have disappeared. We use the regular
    language $\langE= ((\T,\T)+(\N,\N))$ to denote the relation that a process
    is idle. The transition relation
    $\langT$ can be specified as a union of the following regular expressions:
    %in Figure~\ref{fig:tr_herman}.
        (1) $\langE^*$ [{\it Idle}],
        (2) $\langE^*(\T,\N)(\T,\N)\langE^*+(\T,\N)\langE^*(\T,\N)$ 
            [{\it Discard both tokens}], and
        (3) $\langE^*(\T,\N)(\N,\T)\langE^*$+ $(\N,\T)\langE^*(\T,\N)$  
            [{\it Pass the token}].
\end{example}

\vspace{-0.2cm}
\section{Automata Learning}\label{section:learning}
\vspace{-0.2cm}
Suppose $R$ is a regular \emph{target} language whose definition is
not directly accessible.
\emph{Automata learning}
algorithms~\cite{angluin1987learning,rivest:inference1993,kearns:introduction1994,bollig:angluin2009}
automatically infer a~FA~$\cA$ recognising~$R$.
The setting of an online learning algorithm
assumes a~\emph{teacher} who has access to $R$ and can answer the
following two queries: (1) Membership query $\MEM(w)$: is the word $w$ a~member of $R$, i.e., $w \in R$? (2) Equivalence query $\EQ(\cA)$: is the language of FA~$\cA$ equal to
$R$, i.e., $\langA = R$?
If not, it returns a counterexample $w \in \langA \ominus R$.
\hide{
\begin{itemize}
	\item Membership query $\MEM(w)$: is the word $w$ a~member of $R$, i.e., $w \in R$?
	\item Equivalence query $\EQ(\cA)$: is the language of FA~$\cA$ equal to
	$R$, i.e., $\langA = R$?
	If not, what is~a counterexample to this equality, i.e., a word  $w \in \langA \ominus R$?
\end{itemize}
}
The learning algorithm will then construct a~FA $\cA$ such
that $\langA = \langR$ by interacting with the teacher.
Such an algorithm works iteratively: In each iteration, it
performs membership queries to get from the teacher information about $R$.
Using the results of the queries, it proceeds by
constructing a~candidate automaton $\cA_h$ and makes an
equivalence query $\EQ (\cA_h)$. If $\langA_h = \langR$, the algorithm terminates
with $\cA_h$ as the resulting FA.
Otherwise, the teacher returns a word $w$ distinguishing $\langA_h$ from~$\langR$.
The learning algorithm uses $w$ to
refine the candidate automaton of the next iteration. In the last decade, automata learning algorithms have been frequently applied to solve formal verification and synthesis problems, c.f.,~\cite{ChenHLLTWW16,chapman:learning2015,habermehl2005regular,grinchtein2006inferring,ChenFCTW09,FarzanCCTW08}.

More concretely, below we explain the details of the automata learning algorithm proposed by Rivest and Schapire~\cite{rivest:inference1993} (RS), which is an improved version of the classic $L^*$ learning algorithm by Angluin~\cite{angluin1987learning}. 
The foundation of the learning algorithm is the Myhill-Nerode theorem, from which one can infer that the states of the minimal DFA recognizing $\langR$ are isomorphic to the set of equivalence classes defined by the following relations:
$x \equiv_\langR y \mbox{ iff } \forall z\in \Sigma^*: xz\in \langR \leftrightarrow yz\in \langR.$
Informally, two strings $x$ and $y$ belong to the same state of the minimal DFA recognising $\langR$ iff they cannot be distinguished by any suffix $z$. In other words, if one can find a suffix $z'$ such that $xz' \in \langR$ and $yz' \notin \langR$ or vice versa, then $x$ and $y$ belong to different states of the minimal DFA.

The algorithm uses a data structure called {\em observation table} $(S, E, T)$ to find the equivalence classes correspond to $\equiv_\langR$, where $S$ is a set of strings denoting the set of identified states, $E$ is the set of suffixes to distinguish if two strings belong to the same state of the minimal DFA, and $T$ is a mapping from $(S\cup (S\cdot \Sigma))\cdot E$ to $\{\top,\bot\}$. The value of $T(w)=\top$ iff $w\in \langR$.
We use $\mathit{row}_E(x) = \mathit{row}_E(y)$ as a shorthand for $\forall z\in E: T(xz)=T(yz)$. That is, the strings $x$ and $y$ cannot be identified as two different states using only strings in the set $E$ as the suffixes. Observe that $x \equiv_\langR y$ implies $\mathit{row}_E(x) = \mathit{row}_E(y)$ for all $E\subseteq \Sigma^*$.
We say that an observation table is {\em closed} iff $\forall x\in S, a \in \Sigma: \exists y\in S: \mathit{row}_E(xa)=\mathit{row}_E(y)$. 
Informally, with a closed table, every state can find its successors wrt. all symbols in $\Sigma$.
Initially, $S=E=\{\lambda\}$, and $T(w)=\MEM(w)$ for all $w\in \{\lambda\}\cup \Sigma$.

\begin{algorithm}[h]
	\KwIn{A teacher answers $\MEM(w)$ and $\EQ(\cA)$ about a target regular language $R$ and the initial observation table $(S,E,T)$.}
	\Repeat{$\EQ(\cA_h)=\ltrue$}{
	\While{$(S,E,T)$ is not closed}{
		Find a pair $(x,a)\in S\times \Sigma$ such that $\forall y\in S: \mathit{row}_E(xa)\neq \mathit{row}_E(y)$.
		Extend $S$ to $S\cup \{xa\}$ and update $T$ using membership queries accordingly\;
	}
	Build a candidate DFA $\cA_h= (S, \Sigma, \delta, \lambda, F)$, where $\delta=\{(s,a,s')\mid s,s'\in S \wedge \mathit{row}_E(sa)=\mathit{row}_E(s)\}$, the empty string $\lambda$ is the initial state, and $F=\{s\mid T(s)=\top \wedge s\in S \}$\;
	\lIf{$\EQ(\cA_h)=(\lfalse, w)$, where $w \in \langA \ominus \langR$}{
		Analyse $w$ and add a suffix of $w$ to $E$}
	}
	\KwRet $\cA_h$ is the minimal DFA for $R$\;
	\caption{The improved $L^*$ algorithm by Rivest and Schapire}\label{alg:lstar}
\end{algorithm}

The details of of the improved $L^*$ algorithm by Rivest and Schapire can be found in Algorithm~\ref{alg:lstar}.
Observe that, in the algorithm, two strings $x,y$ with $x\equiv_\langR y$ will never be simultaneously contained in the set $S$. When the equivalence query $\EQ(\cA)$ returns $\lfalse$ together with a counterexample $w\in  \langA \ominus \langR$, the algorithm will perform a binary search over $w$ using membership queries to find a suffix $e$ of $w$ and extend $E$ to $E\cup \{e\}$. The suffix $e$ has the property that $\exists x,y\in S, a\in \Sigma: \mathit{row}_E(xa)=\mathit{row}_E(y) \wedge \mathit{row}_{E\cup \{e\}}(xa) \neq \mathit{row}_{E\cup \{e\}}(y)$, that is, add $e$ to $E$ will identify at least one more state. The existence of such a suffix is guaranteed. We refer the readers to~\cite{rivest:inference1993} for the proof.

\begin{proposition}\label{thm:lstar}\cite{rivest:inference1993}
Algorithm~\ref{alg:lstar} will find the minimal DFA $\cR$ for $\langR$ using at most $n$ equivalence queries and $n(n+n|\Sigma|) + n\log m$ membership queries, where $n$ is the number of state of $\cR$ and $m$ is the length of the longest counterexample returned from the teacher.
\end{proposition}
Because each equivalence query with a $\lfalse$ answer will increase the size (number of states) of the candidate DFA by at least one and the size of the candidate DFA is bounded by $n$ according to the Myhill-Nerode theorem, the learning algorithm uses at most $n$ equivalence queries. The number of membership queries required to fill in the entire observation table is bounded by $n(n+n|\Sigma|)$. Since a binary search is used to analyse the counterexample and the number of counterexample from the teacher is bounded by $n$, the number of membership queries required is bounded by $n\log m$.

We would like to introduce the other two important variants of the $L^*$ learning algorithm.
The algorithm proposed by Kearns and Vazirani~\cite{kearns:introduction1994} (KV) uses a {\em classification tree} data structure to replace the observation table data structure of the classic $L^*$ algorithm. The algorithm of Kearns and Vazirani has a similar query complexity to the one of Rivest and Schapire~\cite{rivest:inference1993}; it uses at most $n$ equivalence queries and $n^2(n|\Sigma|+m)$ membership queries. However, the worst case bound of the number of membership queries is very loose. It assumes the structure of the classification tree is linear, i.e., each node has at most one child, which happens very rarely in practice. In our experience, the algorithm of Kearns and Vazirani usually requires a few more equivalence queries, with a significant lower number of membership queries comparing to Rivest and Schapire when applied to verification problems.

The $\mathit{NL}^*$ algorithm~\cite{bollig:angluin2009} learns a non-deterministic finite automaton instead of a deterministic one. More concretely, it makes use of a canonical form of nondeterministic finite automaton, named {\em residual finite-state automaton (RFSA)} to express the target regular language. In some examples, RFSA can be exponentially more succinct than DFA recognising the same languages. In the worst case, the $\mathit{NL}^*$ algorithm uses $O(n^2)$ equivalence queries and $O(m|\Sigma|n^3)$ membership queries to infer a canonical RFSA of the target language.

\hide{
Other online automata learning algorithms have similar guarantee on the required number of queries. That is, they using at most $n$ equivalence queries and a the number of membership queries polynomial

~\cite{angluin1987learning,rivest:inference1993,kearns:introduction1994,bollig:angluin2009}}

%\section{Overview}
%\label{section:overview}

%\vspace{-0.2cm}
\section{Algorithm}
\label{section:algorithm}

%%

%% notes: 1 page
%% -Structure of the paper
%% -The Framework (a figure)

\begin{figure}[tb]
	\centering
%	\resizebox{\columnwidth}{!}{
\scalebox{0.65}{
		\begin{tikzpicture}[punkt/.style={rectangle, rounded corners,
			draw=black, very thick, text width=5em,
			minimum height=8em, text centered}]
		\node[punkt] (learner) {};
		\node[align=left] (learner_text) at ($(learner) + (0, 1cm)$) {Learner};
		\node[punkt, minimum height=12em, text width=11em, right=2.6cm of learner] (teacher) {};
		\node[align=left] (teacher_text) at ($(teacher) + (0, 1.6cm)$) {Teacher};
		\node[align=center, right=0.5cm of teacher] (teacher_right) {$(\cI,\cT,\cB)=\ltrue$ and\\ an inductive\\ invariant $\cA$\\ or\\ $(\cI,\cT,\cB)=\lfalse$ and \\ a word $w\in \langT^*(\langI) \cap \langB$ };
		\node[rectangle, minimum height=2em, rounded corners, draw=black, text centered] (teacher_mem) at ($(teacher)+(0,2.5em)$) {$w\in \langT^*(\langI)?$};
		\node[align=left, rectangle, minimum height=6em, rounded corners, draw=black] (teacher_equ) at ($(teacher)+(0,-2.5em)$) {$(1) \langI\subseteq \langA_h?$\\ $(2) \langA_h\cap \langB=\emptyset?$\\$(3) \langT(\langA_h)\subseteq \langA_h?$};
		
		\draw ($(learner.east)+(0, 3em)$) 
		edge[->] node[above,pos=0.4] { $\MEM(w)$ }
		($(teacher_mem.west)+(0, 0.5em)$);
		\draw ($(learner.east)+(0, 2em)$) 
		edge[<-] node[below,pos=0.4] { $\mathit{yes}$ or $\mathit{no}$ }
		($(teacher_mem.west)+(0, -0.5em)$);
		\draw ($(learner.east)+(0,-2em)$) 
		edge[->] node[above,pos=0.43] { $\EQ(\cA_h)$ }
		($(teacher_equ.west)+(0,+0.5em)$);
		\draw ($(learner.east)+(0,-3em)$) 
		edge[<-] node[below,pos=0.43] { $\lfalse, w$ }
		($(teacher_equ.west)+(0,-0.5em)$);
		\draw ($(teacher_right.west)+(1.5em,0)$)
		edge[<-] 
		($(teacher.east)$) ;
		
		\end{tikzpicture}
	}
	\caption{Overview: using automata learning to solve the regular model checking problem $(\cI,\cT,\cB)$. Recall that we use calligraphy font for automata/transducers and roman font for the corresponding languages/relations.}
	\label{figure:overview}\vspace{0.1cm}
\end{figure}
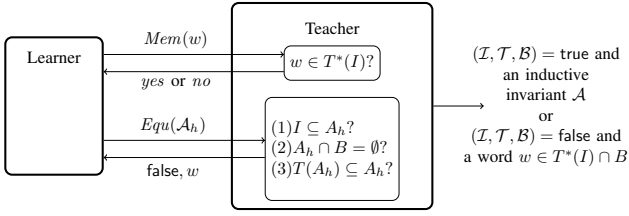
We apply automata learning algorithms, including Angluin's $L^*$ and its variants, to solve the regular model checking problem $(\cI,\cT,\cB)$. Those learning algorithms require a teacher answering both equivalence and membership queries. 
Our strategy is to design a ``strict teacher'' targeting the minimal inductive invariant $\langT^*(\langI)$. For a membership query on a word $w$, the teacher checks if $w\in \langT^*(\langI)$, which is decidable under the assumption that $\cT$ is length-preserving. For an equivalence query on a candidate FA $\cA_h$, the teacher analyses if $\cA_h$ can be used as an inductive invariant in a proof of the problem $(\cI,\cT,\cB)$. It performs one of the following actions depending on the result of the analysis (Fig.~\ref{figure:overview}):
\begin{itemize}
\item Determine that $\cA_h$ does not represent an inductive
  invariant, and return $\lfalse$ together with an explanation $w \in
  \Sigma^*$ to the learner.
\item Conclude that $(\cI,\cT,\cB)=\ltrue$, and terminate the learning
  process with an inductive invariant $\cA_h$ as the proof.
\item Conclude that $(\cI,\cT,\cB)=\lfalse$, and terminate the
  learning with a word $w\in \langT^*(\langI)\cap \langB$ as
  an evidence.
\end{itemize}

Similar to the typical regular model checking approach, our learning-based technique tries to find a ``regular proof'', which amounts to finding an inductive invariant in the form of a regular language. Our approach is incomplete in general since it could happen that there only non-regular inductive invariants exist. Pathological cases where only non-regular inductive invariant exist do not, however, seem to occur frequently in practice, c.f.,~\cite{bouajjani2004abstract,habermehl2005regular,bouajjani2006abstract,RMC,TL10,rmc-thesis,Lin12-fsttcs}.

%The details of the learning based algorithm and its properties will be described in Section~\ref{section:algorithm}. The experimental results will be reported in Section~\ref{section:evaluation}.

%A semi-algorithm for a regular model checking problem $(\cI,\cT,\cB)$ based on automata learning will be presented in this section.
Answering a membership query on a word $w$, i.e., checking whether $w\in \langT^*(\langI)$, is the easy part: since $\cT$ is length-preserving, we can construct an FA recognising $\mathit{Post}^{|w|} = \{w' \mid |w'|=|w| \wedge w'\in \langT^*(\langI)\}$ and then check if $w\in \mathit{Post}^{|w|}$. In practice, $\mathit{Post}^{|w|}$ can be efficiently computed and represented using BDDs and symbolic model checking.

For an equivalence query on a candidate FA $\cA_h$, we need to check if $\cA_h$ can be used as an inductive invariant for the regular model checking problem $(\cI,\cT,\cB)$.
More concretely, we check the three conditions (1) $\langI\subseteq \langA_h$, (2) $\langA_h\cap \langB=\emptyset$, and (3) $\langT(\langA_h)\subseteq \langA_h$ using Algorithm~\ref{alg:equivalence}. 

\begin{algorithm}
	\KwIn{An FA $\cA_h$ and an RMC problem $(\cI,\cT,\cB)$}
	\If{$\langI \not \subseteq \langA_h$}{
		Find a word $w\in \cI \setminus \langA_h$\;
		\KwRet ($\lfalse$, $w$) to the learner\;
	}	
	\ElseIf{$\langA_h \cap \langB \neq \emptyset$}{
		Find a word $w \in \langA_h \cap \langB$\;
		\lIf{$w\in \langT^*(\langI)$}{
			Output \{$cex=w$, $(\cI,\cT,\cB)=\lfalse$\} and halt%
		}
		\lElse{
			\KwRet ($\lfalse$, $w$) to the learner%
		}
	}	
	\ElseIf{$\langT(\langA_h)\not\subseteq \langA_h$}{
		Find a pair of words $(w,w')\in \langT$ such that $w \in \langA_h $ but $ w' \notin  \langA_h$\;
		\lIf{$w\in \langT^*(\langI)$}{
			\KwRet ($\lfalse$, $w'$) to the learner%
		}
		\lElse{
			\KwRet ($\lfalse$, $w$) to the learner%
		}		
	}	
	\lElse{
		Output \{$inv=\cA_h$, $(\cI,\cT,\cB)=\ltrue$\} and halt%
	}
	
	\caption{Answer equivalence query on candidate FA}\label{alg:equivalence}
\end{algorithm}

If the condition (1) is violated, i.e., $\langI \not \subseteq \langA_h$, there is a word $w\in \langI \setminus \langA_h$. Since $\langI\subseteq \langT^*(\langI)$, the teacher can infer that $w \in  \langT^*(\langI) \setminus \langA_h$ and return $w$ as a {\em positive} counterexample to the learner. A counterexample is positive if it represents a word in the target language that was missing in the candidate language. The definition negative counterexamples is symmetric.

If the condition (2) is violated, i.e., $\langA_h \cap \langB \neq \emptyset$, there is a word $w \in \langA_h \cap \langB$. The teacher checks if $w\in \langT^*(\langI)$ by constructing $\mathit{Post}^{|w|}$ and checking if $w\in \mathit{Post}^{|w|}$. If $w\not\in \langT^*(\langI)$, the teacher obtains that $w \in  \langA_h \setminus \langT^*(\langI)$ and returns $\lfalse$ together with $w$ as a negative counterexample to the learner. Otherwise, the teacher infers that $w \in  \langT^*(\langI) \cap \langB$ and outputs $(\cI,\cT,\cB)=\lfalse$ with the word $w$ as an evidence.

The case that the condition (3) is violated, i.e., $\langT(\langA_h)\not\subseteq \langA_h$, is more involved. There exists a pair of words $(w,w')\in \langT$ such that $w \in \langA_h \wedge w' \notin  \langA_h$. The teacher will check if $w\in \langT^*(\langI)$. If it is, then the teacher knows that $w'\in \langT^*(\langI)\wedge w' \notin  \langA_h$ and hence returns $\lfalse$ together with $w'$ as a positive counterexample to the learner.
If $w\notin \langT^*(\langI)$, then the teacher knows that $w\notin \langT^*(\langI)\wedge w \in  \langA_h$ and hence returns $\lfalse$ together with $w$ as a negative counterexample to the learner.

If all conditions hold, the ``strict teacher'' shows its generosity ($\langA_h$ might not equal to $\langT^*(\langI)$, but it will still pass) and concludes that $(\cI,\cT,\cB)=\ltrue$ with a proof using $\cA_h$ as the inductive invariant.

\begin{theorem}[Correctness]
  If the algorithm from Fig.~\ref{figure:overview} terminates,
  it gives correct answer to the RMC problem $(\cI,\cT,\cB)$.
\end{theorem}
%\begin{proof}
%To see this, the output of the algorithm is correct by construction, since 
    To see this, observe that the algorithm provides an inductive invariant 
    when it concludes $(\cI,\cT,\cB)=\ltrue$ and a word in $\langT^*(\langI)\cap \langB$ when it concludes $(\cI,\cT,\cB)=\lfalse$.
%\end{proof}
%
%\vspace{-.1em}
%
    In addition, if one of the $L^*$ learning algorithms\footnote{If $NL^*$ is used, the bound in Theorem~\ref{thm:tem} 
will increase to $O(k^2)$.} from Section~\ref{section:learning} is used, we can obtain an additional
result about termination:
\begin{theorem}[Termination]\label{thm:tem}
	When $\langT^*(\langI)$ is regular, the algorithm from Fig.~\ref{figure:overview} is guaranteed to terminate in at most $k$ iterations, where $k$ is the size of the minimal DFA of $\langT^*(\langI)$.
\end{theorem}
\begin{proof}
%The proof is based on the property of automata learning algorithms, i.e., the learning algorithm will propose a candidate FA $\cA_f$ recognizing $\langT^*(\langI)$ before the $(k+1)$th iteration. Further details are in the appendix.
Observe that in the algorithm, the counterexample obtained by the learner in each iteration locates in the symmetric difference of the candidate language and $\langT^*(\langI)$. Hence, when $\langT^*(\langI)$ can be recognized by a DFA of $k$ states, the algorithm will not execute more than $k$ iterations by Proposition \ref{thm:lstar}. %Further details can be found in the appendix.
\end{proof}

\vspace{-2.9mm}

Two remarks are in order. Firstly, the set $\langT^*(\langI)$ tends to be
regular in practice, e.g.,
see~\cite{bouajjani2004abstract,habermehl2005regular,bouajjani2006abstract,RMC,TL10,rmc-thesis,fast,BFLS05,Lin12-fsttcs,LS05}.
In fact, it is known that $\langT^*(\langI)$ is regular for many subclasses
of infinite-state systems that can be modelled in regular model checking
\cite{TL10,Lin12-fsttcs,Iba78,BFLS05,LS05}
including pushdown systems, reversal-bounded counter systems, two-dimensional
VASS (Vector Addition Systems with States), and other subclasses of counter
systems. Secondly, even in the case when $\langT^*(\langI)$ is not regular,
termination may still happen due to the ``generosity'' of the teacher, which
will accept any inductive invariant as an answer.

\OMIT{
\begin{figure}
\centering
\begin{tabular}{c|c}
	& $\lambda$\\
	\hline
	$\lambda$	     & $\bot$\\
	$\T$				     & $\top$\\
	\hline
	$\N$				     & $\bot$\\
	$\T\T$				& $\bot$\\
	$\T\N$				& $\top$
\end{tabular}
\adjustbox{valign=M}{
\begin{tikzpicture}
\tikzset{every state/.style={minimum size=15pt}}
\node[initial,state] (A) {};
\node[state,accepting] (B) [below =2em of A] {};
\tikzset{loop/.style={in=30,out=330, looseness=5}}
\draw[->,loop right] (A) to node[right] {$\N$} (A);
\tikzset{loop/.style={in=30,out=330, looseness=5}}
\draw[->,loop right] (B) to node[right] {$\N$} (B);
\draw[->, bend left] (A) to node[right]{$\T$}(B);
\draw[->, bend left] (B) to node[left]{$\T$}(A);
\end{tikzpicture}}
\caption{Using learning to solve the RMC problem $(\cI,\cT,\cB)$ of Herman's Protocol. The table on the left is the content of the {\em observation table} used by the automata learning algorithm of Rivest and Shaphire~\cite{rivest:inference1993} and the automaton on the right is the inferred candidate DFA.}\label{figure:learning-token-passing}
\end{figure}
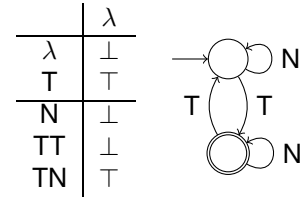

Note, though, that the termination may still happen earlier even in the case 
when $\langT^*(\langI)$ is not regular.
The theorem can be proved by contradiction. Assume that our learning-based algorithm reached iteration $k+1$. It follows that the learner received $k+1$ counterexamples. Observe that all counterexamples returned to the teacher are in the symmetric difference of $\langT^*(\langI)$ and the candidate language $\langA_h$. By the property of automata learning algorithms (check Section~\ref{section:learning}), the learning algorithm will propose a candidate FA $\cA_f$ recognizing $\langT^*(\langI)$ before the $(k+1)$th iteration. Since $\langA_f=\langT^*(\langI)$, it holds that $\langI\subseteq \langA_f$ and $\langT(\langA_f)\subseteq \langA_f$. If $\langA_f \cap \langB =\emptyset$, the teacher will terminate and return $(\cI,\cT,\cB)=\ltrue$ in addition with a proof using $\cA_f$ as the inductive invariant. Otherwise, it will terminate and return $(\cI,\cT,\cB)=\lfalse$ in addition with a word in $\langT^*(\langI) \cap \langB$ as evidence. Then we find a contradiction because the algorithm will terminate at the iteration when an equivalence query on $\cA_f$ was posed to the teacher, which is before the $(k+1)$th iteration.

%\begin{adjustwidth}{}{4cm}
\begin{example}[RMC of the Herman's Protocol]
	Consider the RMC problem of Herman's Protocol in Example~\ref{example:herman}. Initially, several membership queries will be posed to the teacher to produce the closed observation table on the left of Fig.~\ref{figure:learning-token-passing}. In this example, the teacher returns $\top$ only for words containing an odd number of the symbol \T. 
	The learner will then construct the candidate FA $\cA_h$ on the right of Fig.~\ref{figure:learning-token-passing} and pose an equivalence query on $\cA_h$. Observe that $\cA_h$ can be used as an inductive invariant in a regular proof. It is easy to verify that $\langI = \N^*\T(\N^*\T\N^*\T\N^*)^* \subseteq \langA_h$ and $\langA_h \cap \langB =\langA_h\cap \N^* =\emptyset$. The condition (3) $\langT(\langA_h)\subseteq \langA_h$ of a regular proof can be proved to be correct based on the following observation:  the FA $\cA_h$ recognises exactly the set of all configurations with an odd number of tokens. When tokens are discarded in a transition, the total number of discarded tokens in all processes is always $2$. The other two types of transitions will not change the total number of tokens in the system. It follows that taking a transition from any configurations in $\langA_h$ will arrive a configuration with an odd number of tokens, which is still in $\langA_h$.
\end{example}
%\end{adjustwidth}

The verification of Herman's Protocol finishes after the first iteration of learning and hence we cannot see how the learning algorithm uses a counterexample for refinement. Below we introduce a slightly more difficult problem.

Israeli-Jalfon's Protocol is a routing protocol of $n$ processes organised in a ring-shaped topology, where a process may hold a token.
Again we assume the processes are numbered from $0$ to $n-1$. If the process with the number $i$ is chosen by the scheduler, it will toss a coin to decide  to pass the token to the left or right, i.e. the one with the number $(i-1)\%n$ or $(i+1)\%n$. When two tokens are held by the same process, they will be merged. The safety property of interest is that every system configuration has at least one token.
The protocol and the corresponding safety property can be modelled as a regular model checking problem $(\cI, \cT, \cB)$, together with
the set of initial configurations $\langI=(\T+\N)^*\T(\T+\N)^*\T(\T+\N)^*$, i.e., at least two processes have tokens, and the set of bad configurations $\langB=\N^*$, i.e., all tokens have disappeared. Again we use the regular language $\langE= ((\T,\T)+(\N,\N))$ to denote the relation that a process is idle, i.e, the process does not change its state. The transition relation $\langT$ then can be specified as a union of the regular expressions in Figure~\ref{fig:tr_israeli}.

\begin{figure*}[h]
\centering
\begin{tabular}{lcl}
	$\langE^*(\T,\N)((\T,\T)+(\N,\T))\langE^*$, $((\T,\T)+(\N,\T))\langE^*(\T,\N)$  & &({\it Pass the token right})\\
	$\langE^*((\T,\T)+(\N,\T))(\T,\N)\langE^*$, $(\T,\N)\langE^*((\T,\T)+(\N,\T))$  &  &({\it Pass the token left})
\end{tabular}
\caption{The Transition Relation of Israeli-Jalfon's Protocol}
\label{fig:tr_israeli}
\end{figure*}

When the automata learning algorithm of Rivest and Shaphire is applied to solve the RMC problem, we can obtain an inductive invariant with 4 states in 3 iterations.
The first candidate FA in Fig.~\ref{figure:learning-IJ}(a) is incorrect because it does not include the initial configuration \T\T. By analysing the counterexample \T\T, the learning algorithm adds the suffix \T\ to the set $E$. The second candidate FA in Fig.~\ref{figure:learning-IJ}(b) is still incorrect because it contains an unreachable bad configuration \N\N\N. The learning algorithm analyses the counterexample \N\N\N\ and adds the suffix \N\ to the set $E$. This time it obtains the candidate FA  in Fig.~\ref{figure:learning-IJ}(c), which is a valid regular inductive invariant.

\begin{figure*}[tb]
	\begin{minipage}[t]{0.25\textwidth}
		\centering
		\begin{tabular}{c|c}
	& $\lambda$\\
	\hline
	$\lambda$	     & $\bot$\\
	\hline
	$\T$				     & $\bot$\\
	$\N$				     & $\bot$\\
\end{tabular}
\adjustbox{valign=M}{
	\begin{tikzpicture}
	\tikzset{every state/.style={minimum size=15pt}}
	\node[initial above,state] (A) {};
	\tikzset{loop/.style={in=300,out=240, looseness=5}}
	\draw[->,loop right] (A) to node[below] {$\T,\N$} (A);
	\end{tikzpicture}}\\[2ex]
		(a) First candidate automaton
	\end{minipage}
	\begin{minipage}[t]{0.3\textwidth}
		\centering
		\begin{tabular}{c|c|c}
	& $\lambda$&\T\\
	\hline
	$\lambda$	     &$\bot$&$\bot$\\
	$\T$				     &$\bot$&$\top$\\
	$\T\T$				&$\top$&$\top$\\
	\hline
	$\N$				     &$\bot$&$\top$\\
	$\T\N$				&$\top$&$\top$\\
	$\T\T\T$				&$\top$&$\top$\\
	$\T\T\N$				&$\top$&$\top$\\
\end{tabular}
\adjustbox{valign=M}{
	\begin{tikzpicture}
	\tikzset{every state/.style={minimum size=15pt}}
	\node[initial above,state] (A) {};
	\node[state] (B) [below =2em of A] {};
	\node[state,accepting] (C) [below =2em of B] {};
	
	\tikzset{loop/.style={in=300,out=240, looseness=5}}
	\draw[->,loop right] (C) to node[below] {$\T,\N$} (C);
	\draw[->] (A) to node[right]{$\T,\N$}(B);
	\draw[->] (B) to node[right]{$\T,\N$}(C);
\end{tikzpicture}}\\
		(b) Second candidate automaton
	\end{minipage}
	\begin{minipage}[t]{0.45\textwidth}
		\centering
		\begin{tabular}{c|c|c|c}
	& $\lambda$&\T&\N\\
	\hline
	$\lambda$	     &$\bot$&$\bot$&$\bot$\\
	$\T$				     &$\bot$&$\top$&$\top$\\
	$\N$				     &$\bot$&$\top$&$\bot$\\
	$\T\T$				&$\top$&$\top$&$\top$\\
	\hline
	$\T\N$				&$\top$&$\top$&$\top$\\
	$\N\T$				&$\top$&$\top$&$\top$\\
	$\N\N$				&$\top$&$\top$&$\bot$\\
	$\T\T\T$				&$\top$&$\top$&$\top$\\
	$\T\T\N$				&$\top$&$\top$&$\top$\\
\end{tabular}
\adjustbox{valign=M}{
	\begin{tikzpicture}
	\tikzset{every state/.style={minimum size=15pt}}
	\node[initial,state] (A) {};
	\node[state] (B) [below left =2em of A] {};
	\node[state] (C) [below right =2em of A] {};
	\node[state,accepting] (D) [below =4em of A] {};
	
	\tikzset{loop/.style={in=300,out=240, looseness=5}}
	\draw[->,loop right] (D) to node[below] {$\T,\N$} (D);
	\tikzset{loop/.style={in=30,out=330, looseness=5}}
	\draw[->,loop right] (C) to node[right] {$\N$} (C);
	\draw[->] (A) to node[left]{$\T$}(B);
	\draw[->] (A) to node[right]{$\N$}(C);
	\draw[->] (B) to node[left]{$\T,\N$}(D);
	\draw[->] (C) to node[right]{$\T$}(D);

\end{tikzpicture}}\\[2ex]
		(c) Third candidate automaton
	\end{minipage}
	\caption{Using learning to solve the RMC problem $(\cI,\cT,\cB)$ of Israeli-Jalfon's Protocol. The table on the left of each sub-figure is the content of the {\em observation table} used by the automata learning algorithm of Rivest and Shaphire~\cite{rivest:inference1993}.}\label{figure:learning-IJ}
\end{figure*}
\begin{example}[RMC of the Israeli-Jalfon's Protocol]
\end{example}

}

\paragraph*{Considerations on Implementation}
The implementation of the learning-based algorithm is very simple. Since it is based on standard automata learning algorithms and uses only basic automata/transducer operations, one can find existing libraries for them. The implementation only need to take care of how to answer queries. The core of our implementation has only around 150 lines of code (excluding the parser of the input models).
We provide a few suggestions to make the implementation more efficient.
First, each time when an FA recognising $\mathit{Post}^k$ is produced, we store the pair $(k,\mathit{Post}^k)$ in a cache. It can be reused when a query on any word of length $k$ is posed. We can also check if $\mathit{Post}^k\cap \langB=\emptyset$. The algorithm can immediately terminate and return $(\cI,\cT,\cB)=\lfalse$ if $\mathit{Post}^k\cap \langB\neq\emptyset$. Second, for each language inclusion test, if the inclusion does not hold, we suggest to return the shortest counterexample. This heuristic helped to shorten the average length of strings sent for membership queries and hence reduced the cost of answering them. Recall that the algorithm needs to build the FA of $\mathit{Post}^k$ to answer membership queries. The shorter the average length of query strings is, the fewer instances of $\mathit{Post}^k$ have to be built.

\section{Evaluation}
\label{section:evaluation}
% Preview body

To evaluate our techniques, we have developed a prototype\footnote{Available at \url{https://github.com/ericpony/safety-prover}.}
in Java and used the \libalf\ library \cite{bollig2010libalf}
as the default inference engine. We used our prototype to check safety
properties for a range of parameterised systems, including cache coherence
protocols (German's Protocol), self-stabilising protocols
(Israeli-Jalfon's Protocol and Herman's Protocol), synchronisation protocols
(Lehmann-Rabin's Dining Philosopher Protocol), secure multi-party
computation protocol (David Chaums' Dining Cryptographers Protocol),
and mutual exclusion protocols (Szymanski's Protocol, Burn's Protocol,
Dijkstra's Protocol, Lamport's Bakery Algorithm, and Resource-Allocator
Protocol). Most of the examples we consider are standard benchmarks
in the literature of regular model checking (c.f.  \cite{abdulla2007regular,abdulla2013all,bouajjani2006abstract,RMC,rmc-thesis}).
Among them, German's Protocol and Kanban are more difficult than the
other examples for fully automatic verification (c.f. \cite{abdulla2007regular,abdulla2013all,kaiser2010dynamic}).

Based on these examples, we compare our learning method with existing
techniques such as SAT-based learning \cite{Neider-thesis,Neider13,lin2015regular,lin2016liveness},
extrapolating~\cite{tormc,legay2008t}, and abstract regular model checking (ARMC)~\cite{bouajjani2006abstract,bouajjani2004abstract}.
The SAT-based learning approach encodes automata as Boolean formulae and
exploits a SAT-solver to search for candidate automata representing
inductive invariants. It uses automata-based algorithms to either verify the correctness of the
candidate or obtain a counterexample that can be further encoded as a Boolean constraint.
T(O)RMC~\cite{tormc,legay2008t} extrapolates the limit of
the reachable configurations represented by an infinite sequence of automata.
The extrapolation is computed by first identifying the increment between successive automata,
and then over-approximating the repetition of the increment by adding loops to the automata.
ARMC is an efficient technique that integrates abstraction refinement into the fixed-point computation.
It begins with an existential abstraction obtained by merging states in the automata/transducers.
Each time a spurious counterexample is found, the abstraction can be refined by splitting
some of the merged states. ARMC is among the most efficient algorithms for regular model checking~\cite{habermehl2005regular}.

% Preview source code for paragraph 0

\begin{table*}[tb]
\begin{centering}
\scalebox{0.92}{ %
\begin{tabular}{|l|c|c|c|c|c|c|c||c|c|c||c|c|c||c|c|c||c|}
\hline 
\multicolumn{8}{|c||}{\textbf{The RMC problems}} & \multicolumn{3}{c||}{\textbf{RS}} & \multicolumn{3}{c||}{\textbf{SAT}} & \multicolumn{3}{c||}{\textbf{T(O)RMC}} & \textbf{ARMC}\tabularnewline
\hline 
\textsf{Name}  & \textsf{\footnotesize{}{}\#}\textsf{\tiny{}{}label}{\tiny{} } & \textsf{\footnotesize{}{}S}\textsf{\tiny{}{}init}{\tiny{} } & \textsf{\footnotesize{}{}T}\textsf{\tiny{}{}init}{\tiny{} } & \textsf{\footnotesize{}{}S}\textsf{\tiny{}{}trans} & \textsf{\footnotesize{}{}T}\textsf{\tiny{}{}trans} & \textsf{\footnotesize{}{}S}\textsf{\tiny{}{}bad} & \textsf{\footnotesize{}T}\textsf{\tiny{}bad} & \textsf{\footnotesize{}Time} & \textsf{\footnotesize{}S}\textsf{\tiny{}inv} & \textsf{\footnotesize{}T}\textsf{\tiny{}inv} & \textsf{\footnotesize{}Time} & \textsf{\footnotesize{}S}\textsf{\tiny{}inv} & \textsf{\footnotesize{}T}\textsf{\tiny{}inv} & \textsf{\footnotesize{}Time} & \textsf{\footnotesize{}S}\textsf{\tiny{}inv} & \textsf{\footnotesize{}T}\textsf{\tiny{}inv} & \textsf{\footnotesize{}Time}\tabularnewline
\hline 
\hline 
Bakery \cite{Fokkink-book}  & 3  & 3  & 3  & 5  & 19  & 3  & 9  & 0.0s  & 6  & 18  & 0.5s  & 2  & 5  & 0.0s & 6 & 11 & 0.0s\tabularnewline
\hline 
Burns \cite{abdulla2007regular}  & 12  & 3  & 3  & 10  & 125  & 3  & 36  & 0.2s  & 8  & 96  & 1.1s  & 2  & 10  & 0.1s & 7 & 38 & 0.0s\tabularnewline
\hline 
Szymanski \cite{szymanski1988simple}  & 11  & 9  & 9  & 118  & 412  & 13  & 40  & 0.3s  & 43  & 473  & 1.6s  & 2  & 21  & 2.0s & 51 & 102 & 0.1s\tabularnewline
\hline 
German \cite{german1992reasoning}  & 581  & 3  & 3  & 17  & 9.5k  & 4  & 2112  & 4.8s  & 14  & 8134  & t.o.  & \textendash{}  & \textendash{}  & t.o. & \textendash{} & \textendash{} & 10s\tabularnewline
\hline 
Dijkstra \cite{abdulla2007regular}  & 42  & 1  & 1  & 13  & 827  & 3  & 126  & 0.1s  & 9  & 378  & 1.7s  & 2  & 24  & 6.1s & 8 & 83 & 0.3s\tabularnewline
\hline 
Dijkstra, ring \cite{dijkstra1984invariance,fribourg1997reachability} & 12  & 3  & 3  & 13  & 199  & 3  & 36  & 1.4s  & 22  & 264  & 0.9s  & 2  & 14  & t.o. & \textendash{} & \textendash{} & 0.1s\tabularnewline
\hline 
Dining Crypto. \cite{chaum1988dining}  & 14  & 10  & 30  & 17  & 70  & 12  & 70  & 0.1s  & 32  & 448  & t.o.  & \textendash{}  & \textendash{}  & t.o. & \textendash{} & \textendash{} & 7.2s\tabularnewline
\hline 
Coffee Can \cite{lin2016liveness}  & 6  & 8  & 18  & 13  & 34  & 5  & 8  & 0.0s  & 3  & 18  & 0.2s  & 2  & 7  & 0.1s & 6 & 13 & 0.0s\tabularnewline
\hline 
Herman, linear \cite{herman1990probabilistic}  & 2  & 2  & 4  & 4  & 10  & 1  & 1  & 0.0s  & 2  & 4  & 0.2s  & 2  & 4  & 0.0s & 2 & 4 & 0.0s\tabularnewline
\hline 
Herman, ring \cite{herman1990probabilistic}  & 2  & 2  & 4  & 9  & 22  & 1  & 1  & 0.0s  & 2  & 4  & 0.4s  & 2  & 4  & 0.0s & 2 & 4 & 0.0s\tabularnewline
\hline 
Israeli-Jalfon \cite{israeli1990token}  & 2  & 3  & 6  & 24  & 62  & 1  & 1  & 0.0s  & 4  & 8  & 0.1s  & 2  & 4  & 0.0s & 4 & 8 & 0.0s\tabularnewline
\hline 
Lehmann-Rabin \cite{lehmann1981advantages}  & 6  & 4  & 4  & 14  & 96  & 3  & 13  & 0.1s  & 8  & 48  & 0.5s  & 2  & 11  & 0.8s & 19 & 105 & 0.0s\tabularnewline
\hline 
LR Dining Philo. \cite{lin2016liveness}  & 4  & 4  & 4  & 3  & 10  & 3  & 4  & 0.0s  & 4  & 16  & 0.2s  & 2  & 6  & 0.1s & 7 & 18 & 0.0s\tabularnewline
\hline 
Mux Array \cite{fribourg1997reachability}  & 6  & 3  & 3  & 4  & 31  & 3  & 18  & 0.0s  & 5  & 30  & 0.4s  & 2  & 7  & 0.2s & 4 & 14 & 0.0s\tabularnewline
\hline 
Res. Allocator \cite{donaldson2007automatic}  & 3  & 3  & 3  & 7  & 25  & 4  & 9  & 0.0s  & 5  & 15  & 0.0s  & 1  & 3  & 0.0s & 4 & 9 & 0.0s\tabularnewline
\hline 
Kanban \cite{abdulla2013all,kaiser2010dynamic}  & 3  & 25  & 48  & 98  & 250  & 37  & 68  & t.o.  & \textendash{}  & \textendash{}  & t.o.  & \textendash{}  & \textendash{}  & t.o. & \textendash{} & \textendash{} & 3.5s\tabularnewline
\hline 
Water Jugs \cite{waterjug}  & 11  & 5  & 6  & 23  & 132  & 5  & 12  & 0.1s  & 24  & 264  & t.o.  & \textendash{}  & \textendash{}  & t.o. & \textendash{} & \textendash{} & 0.0s\tabularnewline
\hline 
\end{tabular}} 
\par\end{centering}
\vspace{.6em}
\protect\caption{Comparing the performance of different RMC techniques. 
\textsf{\footnotesize{}{}\#}\textsf{\tiny{}{}label} stands for the size of alphabet;
\textsf{\footnotesize{}{}S}\textsf{\tiny{}{}x} and 
\textsf{\footnotesize{}{}T}\textsf{\tiny{}{}x} stand for the numbers of states
and transitions, respectively, in the automata/transducers.
\textbf{RS} is the result of our prototype using Rivest and Schapire's version of $L^{*}$;
\textbf{SAT}, \textbf{T(O)RMC}, and \textbf{ARMC} are the results of the other three techniques.
}
\vspace{-2em}
\label{tab:comp_tools} 
\end{table*}

The comparison of those algorithms are reported in Table~\ref{tab:comp_tools}, running on a MinGW64
system with 3GHz Intel i7 processor, 2GB memory limit, and 60-second
timeout. The experiments show that the learning method is quite efficient:
the results of our prototype are comparable with those of the ARMC algorithm\footnote
{Available at \url{http://www.fit.vutbr.cz/research/groups/verifit/tools/hades}.}
on all examples but Kanban, for which the minimal inductive invariant,
if it is regular, has at least 400 states. On the other hand, our algorithm
is significantly faster than ARMC in two cases, namely German's Protocol
and Dining Cryptographers.
%Such results are very positive taking into account that termination guarantees for ARTMC are not yet clear.
ARMC comes with a bundle of options and heuristics,
but not all of them work for our benchmarks.
We have tested all the heuristics available from the tool
and adopted the ones\footnote{The heuristics are structure preserving,
backward computation, and backward collapsing with all states being predicates. See \cite{bouajjani2004abstract}
for explanations.} that had the best performance in our experiments.
The performance of SAT-based learning is comparable to the previous
two approaches whenever inductive invariants representable by
automata with few states exist. However, as its runtime grows exponentially
with the sizes of candidate automata, the SAT-based algorithm fails to solve
four examples that do not have small regular inductive invariants.
T(O)RMC seems to suffer from similar problems as it timeouts on all examples
that cannot be proved by the SAT-based approach.

Table~\ref{tab:comp_alg} reports the results of the learning-based algorithm geared with different
automata learning algorithms implemented in \textsf{\small{}libalf}.
As the table shows, these algorithms have similar performance on small examples; however, the algorithm of Rivest and Schapire~\cite{rivest:inference1993} and the algorithm of Kearns and Varzirani~\cite{kearns:introduction1994} are
significantly more efficient than the other algorithms on some large examples
such as Szymanski and German.
 Table~\ref{tab:comp_alg} shows that Kearns and Varzirani's algorithm
can often find smaller inductive invariants (fewer states) than the
other $L^*$ variants, which explains the performance difference.
 For $\mathit{NL}^*$, our implementation pays an additional cost to determinise the learned FA in order to answer the equivalence queries; this cost is significant when a large invariant is needed.

Recall that our approach uses a ``strict but generous teacher''. Namely, the target language of the teacher is $\langT^*(\langI)$ for an RMC problem $(\cI,\cT,\cB)$. We have tried the version where a ``flexible and generous teacher'' is used, that is, the target language of the teacher is the complement of $(\langT^{-1})^*(\langB)$. The performance, however, is worse than that of our current version. This result may reflect the fact that the set  $\langT^*(\langI)$ is ``more regular'' (i.e., can be expressed by a DFA with fewer states) than the set $(\langT^{-1})^*(\langB)$ in practical cases.

\begin{table*}[htb]
\centering
	
\begin{tabular}{|l||c|c|c||c|c|c||c|c|c||c|c|c||c|c|c|}
	
\hline 
&\multicolumn{3}{|c||}{$\mathbf{RS}$}&\multicolumn{3}{|c||}{$\mathbf{L^*}$}&\multicolumn{3}{c||}{$\mathbf{L^*c}$}&\multicolumn{3}{c||}{$\mathbf{KV}$}&\multicolumn{3}{c|}{$\mathbf{NL^*}$}\\
\cline{2-16}
 & \textsf{Time} & \textsf{\footnotesize{}S}\textsf{\tiny{}inv} & \textsf{\footnotesize{}T}\textsf{\tiny{}inv} & \textsf{Time}  & \textsf{\footnotesize{}S}\textsf{\tiny{}inv} & \textsf{\footnotesize{}T}\textsf{\tiny{}inv} & \textsf{Time}   & \textsf{\footnotesize{}S}\textsf{\tiny{}inv} & \textsf{\footnotesize{}T}\textsf{\tiny{}inv} & \textsf{Time}  & \textsf{\footnotesize{}S}\textsf{\tiny{}inv} & \textsf{\footnotesize{}T}\textsf{\tiny{}inv} & \textsf{Time}  & \textsf{\footnotesize{}S}\textsf{\tiny{}inv} & \textsf{\footnotesize{}T}\textsf{\tiny{}inv}\tabularnewline
\hline 
\hline 
Bakery & 0.0s & 6 & 18 & 0.0s & 6 & 18 & 0.1s & 6 & 18 & 0.0s & 6 & 18 & 0.1s & 6 & 18\tabularnewline
\hline 
Burns & 0.2s & 8 & 96 & 0.5s & 8 & 96 & 0.2s & 8 & 96 & 0.2s & 8 & 96 & 0.4s & 6 & 72\tabularnewline
\hline 
Szymanski & 0.3s & 43 & 473 & 2.4s & 51 & 561 & 1.2s & 41 & 451 & 0.3s & 41 & 451 & 1.4s & 59 & 649\tabularnewline
\hline 
German & 4.8s & 14 & 8134 & 13s & 15 & 8715 & 26s & 15 & 8715 & 4.2s & 14 & 8134 & 40s & 15 & 8715\tabularnewline
\hline 
Dijkstra & 0.1s & 9 & 378 & 0.4s & 9 & 378 & 0.1s & 9 & 378 & 0.2s & 9 & 378 & 0.2s & 10 & 420\tabularnewline
\hline 
Dijkstra, ring & 1.4s & 22 & 264 & 2.7s & 20 & 240 & 8.9s & 22 & 264 & 1.5s & 14 & 168 & 1.8s & 20 & 240\tabularnewline
\hline 
Dining Crypto. & 0.1s & 32 & 448 & 0.2s & 34 & 476 & 0.2s & 38 & 532 & 0.1s & 19 & 266 & 0.3s & 36 & 504\tabularnewline
\hline 
Coffee Can & 0.0s & 3 & 18 & 0.0s & 3 & 18 & 0.0s & 4 & 24 & 0.0s & 3  & 18 & 0.0s & 4 & 24\tabularnewline
\hline 
Herman, linear & 0.0s & 2 & 4 & 0.0s & 2 & 4 & 0.0s & 2 & 4 & 0.0s & 2 & 4 & 0.0s & 2 & 4\tabularnewline
\hline 
Herman, ring & 0.0s & 2 & 4 & 0.0s & 2 & 4 & 0.0s & 2 & 4 & 0.0s & 2 & 4 & 0.0s & 2 & 4\tabularnewline
\hline 
Israeli-Jalfon & 0.0s & 4 & 8 & 0.0s & 4 & 8 & 0.0s & 4 & 8 & 0.0s & 4 & 8 & 0.0s & 4 & 8\tabularnewline
\hline 
Lehmann-Rabin & 0.1s & 8 & 48 & 0.2s & 8 & 48 & 0.1s & 8 & 48 & 0.1s & 8 & 48 & 0.2s & 8 & 48\tabularnewline
\hline 
LR D. Philo. & 0.0s & 4 & 16 & 0.2s & 4 & 16 & 0.0s & 5 & 20 & 0.0s & 4 & 16 & 0.0s & 8 & 32\tabularnewline
\hline 
Mux Array & 0.0s & 5 & 30 & 0.0s & 5 & 30 & 0.0s & 5 & 30 & 0.0s & 5 & 30 & 0.0s & 5 & 30\tabularnewline
\hline 
Res. Allocator & 0.0s & 5 & 15 & 0.0s & 4 & 12 & 0.0s & 5 & 15 & 0.0s & 5 & 15 & 0.0s & 5 & 15\tabularnewline
\hline 
Kanban & {\small{}\textgreater{}}60s & \textendash{} & \textendash{} & {\small{}\textgreater{}}60s & \textendash{} & \textendash{} & {\small{}\textgreater{}}60s & \textendash{} & \textendash{} & {\small{}\textgreater{}}60s & \textendash{} & \textendash{} & {\small{}\textgreater{}}60s & \textendash{} & \textendash{}\tabularnewline
\hline 
Water Jugs & 0.1s & 24 & 264 & 0.5s & 25 & 275 & 0.5s & 25 & 275 & 0.1s & 24 & 264 & 0.5s & 25 & 275\tabularnewline
\hline 
\end{tabular}
\vspace{.5em}
\caption{Comparing the performance based on different automata learning algorithms. 
The columns 
 $\mathbf{L^*}$, $\mathbf{L^*c}$, $\mathbf{RS}$, $\mathbf{KV}$, and $\mathbf{NL^*}$ are the results of
the original $L^*$ algorithm by Angluin \cite{angluin1987learning},
a variant of $L^*$ that adds all suffixes of the counterexample to columns,
the version by Rivest and Shapire~\cite{rivest:inference1993},
the version by Kearns and Vazirani~\cite{kearns:introduction1994},
and the $\mathit{NL}^*$ algorithm~\cite{bollig:angluin2009}, respectively.
}\label{tab:comp_alg}
\vspace{-2em}
\end{table*}

\section{Conclusion}
\label{section:conclusion}
The encouraging experimental results suggest that the performance of the $L^*$
algorithm for synthesising regular inductive invariants
is comparable to the most sophisticated algorithm for regular model 
checking for proving safety.
From a theoretical viewpoint, learning-based approaches (including
ours and \cite{Neider-thesis,Neider13,habermehl2005regular})
have a termination guarantee when the set $\langT^*(\langI)$ is regular, which is not guaranteed by approaches based on a fixed-point computation (e.g.,  the ARMC~\cite{bouajjani2004abstract}).
An interesting research question is whether $L^*$ algorithm can be effectively
used for verifying other properties, e.g., liveness.
%In addition with the fact that the proposed approach is extremely simple, we do believe that it has tremendous potential in solving regular model checking problems.

\paragraph*{Acknowledgements}

This article is a full version of \cite{CHLR17}. We thank anonymous referees 
for their useful comments. R\"{u}mmer was supported by the Swedish
Research Council under grant~2014-5484.

\bibliographystyle{abbrv}
\bibliography{refs}

\newpage
\section{Appendix}
We provide some more examples for regular model checking in the appendix.
\OMIT{
The appendix consists of the complete proof of Theorem~\ref{thm:tem} and more examples.
\begingroup
\def\thetheorem{\ref{thm:tem}}

\begin{theorem}[Termination]
	When $\langT^*(\langI)$ is regular, our learning-based algorithm is guaranteed to terminate in at most $k$ iterations, where $k$ is the size of the minimal DFA of $\langT^*(\langI)$.
\end{theorem}
\begin{proof}
	The theorem is proved by contradiction. Assume that our learning-based algorithm reached iteration $k+1$. The learner received $k+1$ counterexamples. Observe that all counterexamples returned by the teacher are in the symmetric difference of $\langT^*(\langI)$ and the candidate language $\langA_h$. By the property of automata learning algorithms (check Section~\ref{section:learning}), the learning algorithm will propose a candidate FA $\cA_f$ recognizing $\langT^*(\langI)$ before the $(k+1)$th iteration. Since $\langA_f=\langT^*(\langI)$, it holds that $\langI\subseteq \langA_f$ and $\langT(\langA_f)\subseteq \langA_f$. If $\langA_f \cap \langB =\emptyset$, the teacher will terminate and output $(\cI,\cT,\cB)=\ltrue$ in addition with a proof using $\cA_f$ as the inductive invariant. Otherwise, it will terminate and output $(\cI,\cT,\cB)=\lfalse$ in addition with a word in $\langT^*(\langI) \cap \langB$ as evidence. Then we find a contradiction because the algorithm will terminate at the iteration when an equivalence query on $\cA_f$ was posed to the teacher, which is before the $(k+1)$th iteration.
\end{proof}
\addtocounter{theorem}{-1}
\endgroup
}

\begin{figure}
\centering
\begin{tabular}{c|c}
	& $\lambda$\\
	\hline
	$\lambda$	     & $\bot$\\
	$\T$				     & $\top$\\
	\hline
	$\N$				     & $\bot$\\
	$\T\T$				& $\bot$\\
	$\T\N$				& $\top$
\end{tabular}
\adjustbox{valign=M}{
\begin{tikzpicture}
\tikzset{every state/.style={minimum size=15pt}}
\node[initial,state] (A) {};
\node[state,accepting] (B) [below =2em of A] {};
\tikzset{loop/.style={in=30,out=330, looseness=5}}
\draw[->,loop right] (A) to node[right] {$\N$} (A);
\tikzset{loop/.style={in=30,out=330, looseness=5}}
\draw[->,loop right] (B) to node[right] {$\N$} (B);
\draw[->, bend left] (A) to node[right]{$\T$}(B);
\draw[->, bend left] (B) to node[left]{$\T$}(A);
\end{tikzpicture}}
\caption{Using learning to solve the RMC problem $(\cI,\cT,\cB)$ of Herman's Protocol. The table on the left is the content of the {\em observation table} used by the automata learning algorithm of Rivest and Shaphire~\cite{rivest:inference1993} and the automaton on the right is the inferred candidate DFA.}\label{figure:learning-token-passing}
\end{figure}

%\begin{adjustwidth}{}{4cm}
\begin{example}[RMC of the Herman's Protocol]
	Consider the RMC problem of Herman's Protocol in Example~\ref{example:herman}. Initially, several membership queries will be posed to the teacher to produce the closed observation table on the left of Fig.~\ref{figure:learning-token-passing}. In this example, the teacher returns $\top$ only for words containing an odd number of the symbol \T. 
	The learner will then construct the candidate FA $\cA_h$ on the right of Fig.~\ref{figure:learning-token-passing} and pose an equivalence query on $\cA_h$. Observe that $\cA_h$ can be used as an inductive invariant in a regular proof. It is easy to verify that $\langI = \N^*\T(\N^*\T\N^*\T\N^*)^* \subseteq \langA_h$ and $\langA_h \cap \langB =\langA_h\cap \N^* =\emptyset$. The condition (3) $\langT(\langA_h)\subseteq \langA_h$ of a regular proof can be proved to be correct based on the following observation:  the FA $\cA_h$ recognises exactly the set of all configurations with an odd number of tokens. When tokens are discarded in a transition, the total number of discarded tokens in all processes is always $2$. The other two types of transitions will not change the total number of tokens in the system. It follows that taking a transition from any configurations in $\langA_h$ will arrive a configuration with an odd number of tokens, which is still in $\langA_h$.
\end{example}
%\end{adjustwidth}

The verification of Herman's Protocol finishes after the first iteration of learning and hence we cannot see how the learning algorithm uses a counterexample for refinement. Below we introduce a slightly more difficult problem.

\begin{example}[RMC of the Israeli-Jalfon's Protocol]
Israeli-Jalfon's Protocol is a routing protocol of $n$ processes organised in a ring-shaped topology, where a process may hold a token.
Again we assume the processes are numbered from $0$ to $n-1$. If the process with the number $i$ is chosen by the scheduler, it will toss a coin to decide  to pass the token to the left or right, i.e. the one with the number $(i-1)\%n$ or $(i+1)\%n$. When two tokens are held by the same process, they will be merged. The safety property of interest is that every system configuration has at least one token.
The protocol and the corresponding safety property can be modelled as a regular model checking problem $(\cI, \cT, \cB)$, together with
the set of initial configurations $\langI=(\T+\N)^*\T(\T+\N)^*\T(\T+\N)^*$, i.e., at least two processes have tokens, and the set of bad configurations $\langB=\N^*$, i.e., all tokens have disappeared. Again we use the regular language $\langE= ((\T,\T)+(\N,\N))$ to denote the relation that a process is idle, i.e, the process does not change its state. The transition relation $\langT$ then can be specified as a union of the regular expressions in Figure~\ref{fig:tr_israeli}.

\begin{figure}[h]
\centering
\scalebox{0.7}{
\begin{tabular}{lcl}
	$\langE^*(\T,\N)((\T,\T)+(\N,\T))\langE^*$, $((\T,\T)+(\N,\T))\langE^*(\T,\N)$  & &({\it Pass the token right})\\
	$\langE^*((\T,\T)+(\N,\T))(\T,\N)\langE^*$, $(\T,\N)\langE^*((\T,\T)+(\N,\T))$  &  &({\it Pass the token left})
\end{tabular}}
\caption{The Transition Relation of Israeli-Jalfon's Protocol}
\label{fig:tr_israeli}
\end{figure}

When the automata learning algorithm of Rivest and Shaphire is applied to solve the RMC problem, we can obtain an inductive invariant with 4 states in 3 iterations.
The first candidate FA in Fig.~\ref{figure:learning-IJ}(a) is incorrect because it does not include the initial configuration \T\T. By analysing the counterexample \T\T, the learning algorithm adds the suffix \T\ to the set $E$. The second candidate FA in Fig.~\ref{figure:learning-IJ}(b) is still incorrect because it contains an unreachable bad configuration \N\N\N. The learning algorithm analyses the counterexample \N\N\N\ and adds the suffix \N\ to the set $E$. This time it obtains the candidate FA  in Fig.~\ref{figure:learning-IJ}(c), which is a valid regular inductive invariant.
\end{example}

\begin{figure}[htb]
	\begin{minipage}[t]{0.5\textwidth}
		\centering
		\begin{tabular}{c|c}
	& $\lambda$\\
	\hline
	$\lambda$	     & $\bot$\\
	\hline
	$\T$				     & $\bot$\\
	$\N$				     & $\bot$\\
\end{tabular}
\adjustbox{valign=M}{
	\begin{tikzpicture}
	\tikzset{every state/.style={minimum size=15pt}}
	\node[initial above,state] (A) {};
	\tikzset{loop/.style={in=300,out=240, looseness=5}}
	\draw[->,loop right] (A) to node[below] {$\T,\N$} (A);
	\end{tikzpicture}}\\[2ex]
		(a) First candidate automaton
	\end{minipage}\\\ \\
	\begin{minipage}[t]{0.5\textwidth}
		\centering
		\begin{tabular}{c|c|c}
	& $\lambda$&\T\\
	\hline
	$\lambda$	     &$\bot$&$\bot$\\
	$\T$				     &$\bot$&$\top$\\
	$\T\T$				&$\top$&$\top$\\
	\hline
	$\N$				     &$\bot$&$\top$\\
	$\T\N$				&$\top$&$\top$\\
	$\T\T\T$				&$\top$&$\top$\\
	$\T\T\N$				&$\top$&$\top$\\
\end{tabular}
\adjustbox{valign=M}{
	\begin{tikzpicture}
	\tikzset{every state/.style={minimum size=15pt}}
	\node[initial above,state] (A) {};
	\node[state] (B) [below =2em of A] {};
	\node[state,accepting] (C) [below =2em of B] {};
	
	\tikzset{loop/.style={in=300,out=240, looseness=5}}
	\draw[->,loop right] (C) to node[below] {$\T,\N$} (C);
	\draw[->] (A) to node[right]{$\T,\N$}(B);
	\draw[->] (B) to node[right]{$\T,\N$}(C);
\end{tikzpicture}}\\
		(b) Second candidate automaton
	\end{minipage}\\\ \\\ \\
	\begin{minipage}[t]{0.5\textwidth}
		\centering
		\begin{tabular}{c|c|c|c}
	& $\lambda$&\T&\N\\
	\hline
	$\lambda$	     &$\bot$&$\bot$&$\bot$\\
	$\T$				     &$\bot$&$\top$&$\top$\\
	$\N$				     &$\bot$&$\top$&$\bot$\\
	$\T\T$				&$\top$&$\top$&$\top$\\
	\hline
	$\T\N$				&$\top$&$\top$&$\top$\\
	$\N\T$				&$\top$&$\top$&$\top$\\
	$\N\N$				&$\top$&$\top$&$\bot$\\
	$\T\T\T$				&$\top$&$\top$&$\top$\\
	$\T\T\N$				&$\top$&$\top$&$\top$\\
\end{tabular}
\adjustbox{valign=M}{
	\begin{tikzpicture}
	\tikzset{every state/.style={minimum size=15pt}}
	\node[initial,state] (A) {};
	\node[state] (B) [below left =2em of A] {};
	\node[state] (C) [below right =2em of A] {};
	\node[state,accepting] (D) [below =4em of A] {};
	
	\tikzset{loop/.style={in=300,out=240, looseness=5}}
	\draw[->,loop right] (D) to node[below] {$\T,\N$} (D);
	\tikzset{loop/.style={in=30,out=330, looseness=5}}
	\draw[->,loop right] (C) to node[right] {$\N$} (C);
	\draw[->] (A) to node[left]{$\T$}(B);
	\draw[->] (A) to node[right]{$\N$}(C);
	\draw[->] (B) to node[left]{$\T,\N$}(D);
	\draw[->] (C) to node[right]{$\T$}(D);

\end{tikzpicture}}\\[2ex]
		(c) Third candidate automaton
	\end{minipage}
	\caption{Using learning to solve the RMC problem $(\cI,\cT,\cB)$ of Israeli-Jalfon's Protocol. The table on the left of each sub-figure is the content of the {\em observation table} used by the automata learning algorithm of Rivest and Shaphire~\cite{rivest:inference1993}.}\label{figure:learning-IJ}
\end{figure}

\OMIT{
When the set of reachable configurations is not regular, 
our learning-based algorithm is not guaranteed to find an 
inductive invariant even when one exists. We demonstrate
this fact in the following example.
\begin{example}[RMC of a token passing protocol]
Consider an RMC problem $(\cI,\cT,\cB)$ with $\langI=\N\T(\N\T)^*$,
$\langT=\langE^*(\T,\N)(\N,\T)\langE^*$, and $\langB = \N^*$
(recall that $\langE$ denotes $(\T,\T)+(\N,\N)$).
This problem describes a token passing protocol where each process
tries to pass its token to the right when its right process does 
not hold a token. The set of bad configurations,
where no process holds a token, is unreachable since there 
exists at least one token at the beginning and the number of tokens
does not change during execution. Observe that $\langT^*(\langI)$ is 
not regular, for $\langT^*(\langI) \cap \N^*\T^* = \{\N^n\T^n : n\ge 1\}$,
a well-known non-regular set. However, $(\cI,\cT,\cB)=\ltrue$ because
$\N^*\T(\N+\T)^*$ is an inductive invariant.
Unfortunately, our learning-based algorithm could not terminate on 
this example, producing a pathological sequence of counterexamples 
$\N\T$, $\T\N\T$, $\N\N\T\T$, $\N\N\N\T\T\T$, $\N\N\N\N\T\T\T\T$, $\N\N\N\N\N\T\T\T\T\T$...
for all learning algorithms considered in our experiments, 
namely, $L^*$, $L^*c$, $NL^*$, $KV$ and $RS$.
It is still unclear to us whether it is possible to make our algorithm
converge with cleverly chosen counterexamples.
\end{example}
}

\end{document}